\documentclass[submission,copyright,creativecommons,sharealike]{eptcs}

\providecommand{\event}{EXPRESS 2011}

\usepackage{stmaryrd,amsmath,amssymb}
\usepackage{theorem}
\usepackage[english]{babel}

\usepackage{breakurl}

\usepackage{mfsos}
\usepackage{corecifops}
\newcommand{\CIF}{{CIF}}

\newcommand{\syncfs}{\ensuremath{\gamma}}

\usepackage{graphicx}
\usepackage{tikz}
\usetikzlibrary{snakes}
\usetikzlibrary{shapes,arrows,automata}
\usetikzlibrary{shapes.symbols}
\usetikzlibrary{shapes.misc}
\usetikzlibrary{calc}
\usetikzlibrary{fit}

\tikzset{
  state/.style={
    circle,
    draw=black,
    minimum height=5ex
  }
}

\tikzset{
  superstate/.style={
    rectangle,
    draw=black,
    rounded corners=1ex,
    minimum height=10ex,
    minimum width=10ex
  }
}

\tikzset{
  namebox/.style={
    rectangle,
    draw=black,
    rounded corners=0ex,
    minimum height=2ex,
    minimum width=4ex
  }
}

\tikzset{
  autbox/.style={
    rectangle,
    draw=black,
    rounded corners=0pt,
    minimum height=2ex,
    minimum width=4ex
  }
}

\tikzset{
  divline/.style={
    -,
    draw=black,
    loosely dashed
  }
}

\tikzset{
  andbox/.style={
    loosely dashed,
    draw=black,
    rounded corners=0pt,
    minimum height=2ex,
    minimum width=4ex
  }
}

\usepackage{color}

\renewcommand{\concat}{\ensuremath{+\hspace{-1.5ex}+\hspace{0.5ex}}}

\newcommand{\sstate}[1]{\ensuremath{\langle #1 \rangle}}
\newcommand{\satrans}[1]{\ensuremath{\trans{#1}}}

\newcommand{\setrans}[1]{\ensuremath{\etrans{#1}}}

\newcommand{\saof}[0]{
\ensuremath{\stackrel {\mathrm{sync}} {\leadsto}}
}

\newcommand{\compsof}[0]{
\ensuremath{\stackrel {\mathrm{\#}} {\leadsto}}
}

\newcommand{\ipredof}[0]{
\ensuremath{\stackrel {\mathrm{ipred}} {\leadsto}}
}

\newcommand{\invof}[0]{
\ensuremath{\stackrel {\mathrm{inv}} {\leadsto}}
}

\newcommand{\lfunfs}{\ensuremath{\mathrm{L}}}
\newcommand{\lfun}[1]{\ensuremath{\lfunfs(#1)}}

\newtheorem{definition}{Definition}
\newtheorem{theorem}{Theorem}
\newtheorem{prop}{Property}
\newtheorem{proof}{Proof}

\title{Linearization of CIF Through SOS}
\author{
  D.E. Nadales Agut \qquad\qquad M.A. Reniers
  \institute{
    Systems Engineering \\
    Department of Mechanical Engineering \\
    Eindhoven University of Technology (TU/e)
  }
  \email{\{d.e.nadales.agut, m.a.reniers\}@tue.nl}
}

\def\titlerunning{Linearization Through SOS}
\def\authorrunning{D.E. Nadales Agut \& M.A. Reniers}
\date{\today}

\begin{document}

\maketitle

\begin{abstract}
  %% Motivation
  Linearization is the procedure of rewriting a process term into a
  linear form, which consist only of basic operators of the process
  language.
  This procedure is interesting both from a theoretical and a
  practical point of view.
  % Say we are linearizing CIF
  In particular, a linearization algorithm is needed for the
  Compositional Interchange Format (CIF), an automaton based
  modeling language.

  The problem of devising efficient linearization algorithms is not
  trivial, and has been already addressed in literature.
  However, the linearization algorithms obtained are the result of an
  inventive process, and the proof of correctness comes as an
  afterthought. Furthermore, the semantic specification of the
  language does not play an important role on the design of the
  algorithm.

  In this work we present a method for obtaining an efficient
  linearization algorithm, through a step-wise refinement of the SOS
  rules of CIF. As a result, we show how the semantic specification of
  the language can guide the implementation of such a procedure,
  yielding a simple proof of correctness.

\end{abstract}

%% Done!
% \begin{enumerate}
% \item \todo{Replace $\predicates_r$ with $\predicates$?}
% \end{enumerate}

\section{Introduction}
\label{sec:introduction}

%%% Background
% What is linearization?
Linearization is the procedure of rewriting a process term into a
linear form, which consist only of \emph{basic operators} of a process
language~\cite{Khadim3:LinHChiPCTech07,BrandReniCuij:LinHybProcJLAP06,Usenko:LinearizationMuCRLPhD02}. Linearization is also referred to as \emph{elimination} in
ACP style process algebras~\cite{BaetenBastenReniers:ProcessAlgebra10}.

% Why is linearization important? (in general and for CIF)
From a theoretical perspective, linearization of process terms is an
interesting result. It allows to get a better understanding about the
expressiveness of the language constructs, since it shows that all its
terms are reducible to some normal form (which contains only a limited
set of operators of the language). Also, linearization is useful in
proving properties about closed terms, since the number of cases that
needs to be dealt with in a proof by structural
induction becomes smaller.

The Compositional Interchange Format
(CIF) \cite{Baeten7:MultiformCIF2D112Tech10}, is a language for
modeling real-time, hybrid and embedded systems. CIF is developed to establish
inter-operability of a wide range of tools by means of model
transformations to and from the CIF. As such it plays a central role in the
European projects Multiform~\cite{MULTIFORM:MULTIFORMMisc08},
HYCON~\cite{HYCONNoE:HYCONMisc05}, C4C~\cite{C4C:C4CMisc08}, and HYCON
2~\cite{HYCONII11}. \CIF\
has a formal semantics \cite{Baeten7:MultiformCIF2D112Tech10}, which
is defined in terms of Structured Operational Semantics Rules (SOS) in
the style of Plotkin~\cite{PlotkinSOSArtJLAP04}.

Besides its theoretical importance, linearization of \CIF\ models
eliminates operators, such as urgency, that cannot be handled in other
languages.
Since \CIF\ is meant to be used as an interchange format, the
elimination of the operators broadens the set of models that can be
translated to other languages.
For the hierarchical extension of \CIF~\cite{BeoharNadales:hCIF10}, hCIF,
linearization makes the elimination of hierarchy possible, and thus,
all the tools available for CIF become available for use with hCIF models as
well.

%%% Problem to be solved + Justification of its importance
%%% + Related work (why it didn't solve the problem?)
It is our goal to build a linearization algorithm for \CIF, which
results in an efficient representation of the original model,
and such that all the operators of the language, such as
  parallel composition or synchronization are eliminated. The
problem of efficient linearization has been already studied in
literature~\cite{Usenko:LinearizationMuCRLPhD02,
  BrandReniCuij:LinHybProcJLAP06, Khadim3:LinHChiPCTech07} for
process-algebraic languages for describing and analyzing
discrete-event systems and hybrid systems.
However, in the previous cases, the linearization algorithm is the
result of an inventive process, and the proof of correctness comes as
an afterthought. The semantic specification of the language does not
play an important
role on the design of the algorithm.

Previously, we studied the problem of implementing a simulator from
the SOS specification of \CIF~\cite{NadalesReniers:DerivingSimulator11}.
The semantics of \CIF\ is defined in terms of SOS rules,
  which induce a hybrid transition system, where each state contains a
  \CIF\ term followed by a valuation (assignment of values to
  variables). This kind of semantics, even though useful for
  specification purposes, was not suitable for the implementation of a
  simulator (interpreter) for the language.
This problem was solved by giving a set of SOS rules,
called \emph{symbolic rules}, which induced transition
systems that do not contain the valuation part.
It was also
noted that the symbolic transition system induced by these rules is
finite, and it resembles a (CIF) automaton. Thus, the symbolic SOS
rules for \CIF\ offer a straightforward algorithm for linearizing
\CIF\ models. However, the resulting automaton has a size that may be
exponential in the size of the input model.

%%% Research hypothesis
In this work we study the possibility of reusing the existing results on
efficient linearization algorithms for obtaining a linear form of
\CIF\ from
SOS rules.
%%% A sketch of the proposed solution
The idea is to give a more concrete version of the symbolic SOS rules
of \CIF\ (which is in turn a concrete version of the SOS rules with
data), such that the transition system they induce can be translated
to an automaton whose size does not grow exponentially as the result
of interleaving actions (for synchronizing action the growth is still
exponential, but in practice this is not a serious limitation since
synchronization takes place only among a limited number of
components).

%%% Contribution
As a result, we show a linearization procedure, which is obtained
from the SOS specification of the language. In this way, the design of
the algorithm requires less invention steps, reducing the
opportunities to introduce mistakes, and at the same time it yields a
simple proof of correctness.

\section{Setting the Scene}
\label{sec:setting-scene}

For the discussion presented here, we consider a simplified version of
\CIF, which is untimed and contains only automata, a parallel
composition
operator, and a synchronizing action operator. This helps to keep the focus on the ideas, without
distracting the reader with the complexity of \CIF\footnote{This
  language contains over 30 deduction rules}. The techniques and
results presented here can be easily extended to the setting of timed
and hybrid
systems, since we handle concepts such as invariants and
time-can-progress
 conditions in a symbolic manner.

We begin by defining automata and the terms of our language.
Throughout this work, notation $\predicates$ is used to refer to a set
of predicates, $\variables$ is a set of variables, $\actions$ is a
set of actions, $\tau$ is the silent action ($\tau \notin \actions$),
and $\actionstau \triangleq \actions \cup \{\tau\}$.

\begin{definition}[Automaton]
  An automaton is a tuple $(V, \actv, \inv, E, \actS)$, where $V
  \subseteq \locations$ is a set of locations, $\actv \in V
  \rightarrow \predicates$ is the initial predicate function, $\inv
  \in V \rightarrow \predicates$ is the invariant function, $E
  \subseteq V \times \actionstau \times \predicates \times V$ is the
  set of edges, and  $\actS \subseteq \actions$ is a set of
    synchronizing actions.
\end{definition}

Figure~\ref{fig:gate} presents a model of a railroad gate.
% Locations:
It has two modes of operation (locations), closed and opened, denoted
$C$ and $O$ respectively.
% Initial predicates:
Its initial predicate function associates the condition $\mathit{wq} =
[\ ]$ to location $C$ (represented graphically with an incoming arrow
without source location), and the predicate $\false$ to location $O$
(represented by the absence of such an arrow). Here $\mathit{wq}$ is
the waiting queue that contains the id's of the trains waiting to pass
through the gate, $[\ ]$ is the empty list, and we denote lists
by writing their elements between brackets, and separated by commas.
% Invariants:
Location $C$ has $n=0$ as invariant, where $n$ is the numbers of
trains crossing the gate, and location $O$ has invariant $n \leq 1$.
% Actions:
The automaton synchronizes with other components in actions
$\mathit{rq}$, $\mathit{go}$, and $\mathit{out}$.

% Edges
The automaton has four edges. Two edges $(C, \mathit{rq},
\mathit{wq}^+ = \mathit{wq} \concat [\mathit{id^+}], C)$, and $(O,
\mathit{rq}, \mathit{wq}^+ = \mathit{wq} \concat [\mathit{id^+}], O)$,
which are used to enqueue requests from the trains that want to pass
the gate. Given two sequences $\xs x$ and $\xs y$, $\xs x \concat \xs
y$ denotes their concatenation. The predicate $\mathit{wq}^+ = \mathit{wq}
\concat [\mathit{id^+}]$ expresses that the new value of the waiting queue after
performing action $\mathit{rq}$ will be the old waiting queue
($\mathit{wq}$) extended with the id of the train that request access
(this id is contained in variable $id^+$). Graphically these edges are
represented by two self loops in locations $C$ and $O$, labeled
$\mathit{rq}, \mathit{wq}^+ = \mathit{wq} \concat [\mathit{id^+}]$.
The gate can make a transition from the closed state to the opened
state, by issuing a $\mathit{go}$ action, which sends the id at the
front of the waiting queue using variable $p$.

\begin{figure}[htb]
  \centering
  \scalebox{0.9}{
    \begin{tikzpicture}[->,>=stealth',shorten >=1pt,auto, semithick,
  initial text=, node distance=13em]

  %% First, draw the automaton
  \node[state] (closed) {$
    \begin{array}{c}
      \mathit{C} \\
      \inv: n = 0
    \end{array}
    $};
  \node (i) at ($(closed.west)+(-1.5,0)$) {$\mathit{wq}=[\ ]$};

  \node [state] (opened) [right of=closed] {$
    \begin{array}{c}
      \mathit{O}\\
      \inv: n \leq 1
    \end{array}
    $};

  \path (i) edge (closed);

  \path (closed) edge [bend left] node [above]
  {$\mathit{go}, [p^+] \concat \mathit{wq}^+ = \mathit{wq}$} (opened);

  \path (opened) edge [bend left] node [below]
  {$\mathit{out}$} (closed);

  \path (opened) edge [loop below] node [below] {$\mathit{rq},
    \mathit{wq}^+ = \mathit{wq} \concat [\mathit{id^+}]$} (opened);

  \path (closed) edge [loop below] node [below] {$\mathit{rq},
    \mathit{wq}^+ = \mathit{wq} \concat [\mathit{id^+}]$} (closed);

  %% Second, draw the surrounding boxes containing the automaton name
  %% and its declarations.

  % Make the big box
  \node[namebox, fit=(closed) (opened) (current bounding box.south west)
  (current bounding box.north east)] (bigsquare) {};

  % The name box
  \node[namebox,anchor=south west] at (current bounding box.north
  west) {$\mathit{Gate}$};

  % And the declarations box
  \node[namebox,anchor=south east] at
  (bigsquare.north east)
  {$
    \actS = \{\mathit{rq}, \mathit{go}, \mathit{out} \}
    $
  };
\end{tikzpicture}
%%% Local Variables:
%%% mode: latex
%%% TeX-master: "../linearization_through_SOS"
%%% End:
  }
  \caption{\CIF\ model of a gate.}
  \label{fig:gate}
\end{figure}
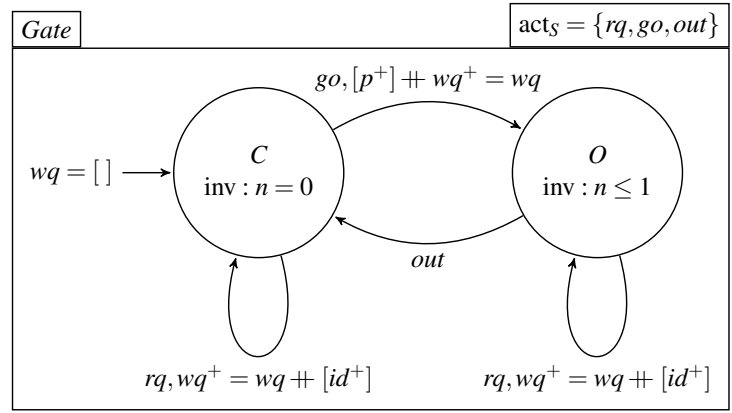

In Figure~\ref{fig:train} we present the model of a train, which will
be run in parallel with the gate model. It has a parameter $i$, which
represents the train's id. It has four locations: far ($F$),
near ($N$), stopped ($S$), and passing ($P$). Location $F$ is the only
initial location. When the train approaches the gate it issues a
request to pass the gate by sending its id though variable $id$. Once
in the near location, it can only go to the passing state if variable
$p$ is updated to its id (this update is carried out by the gate,
as we have seen above). Otherwise it makes a transition to the
stopped state. When the train enters the gate it increments variable
$n$, and it decrements it upon departure.

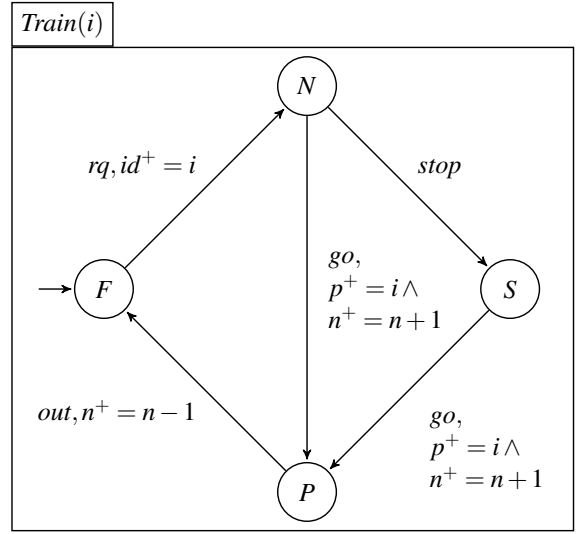
\begin{figure}[htb]
  \centering
  \scalebox{0.9}{
    \begin{tikzpicture}[->,>=stealth',shorten >=1pt,auto, semithick,
  initial text=, node distance=13em]

  %% First, draw the automaton
  \node[initial,state] (safe) {$\mathit{F}$};

  \node [state] (near) at ($(safe.center)+(3,3)$) {$\mathit{N}$};

  \node [state] (stopped) at ($(safe.center)+(6,0)$) {$\mathit{S}$};

  \node [state] (crossing) at ($(safe.center)+(3,-3)$) {$\mathit{P}$};

  \path (safe) edge node
  {$\textit{rq}, \mathit{id}^+ = i$} (near);

  \path (near) edge node
  {$\mathit{stop}$} (stopped);

  \path (near) edge node
  {$
    \begin{array}{l}
      \mathit{go},\\
      p^+=i \wedge{} \\ n^+=n+1
    \end{array}
    $} (crossing);

  \path (stopped) edge node
  {$
    \begin{array}{l}
      \mathit{go},\\
      p^+=i \wedge{} \\ n^+=n+1
    \end{array}
    $} (crossing);

  \path (crossing) edge node {$
    \mathit{out}, n^+ = n - 1
    $} (safe);
  %% Second, draw the surrounding boxes containing the automaton name
  %% and its declarations.

  % Make the big box
  \node[namebox, fit=(safe) (near) (current bounding box.south west)
  (current bounding box.north east)] (bigsquare) {};

  % The name box
  \node[namebox,anchor=south west] at (current bounding box.north
  west) {$\mathit{Train(i)}$};

  % % And the declarations box
  % \node[namebox,anchor=south east] at
  % (bigsquare.north east)
  % {$
  %   \begin{array}{l}
  %     a\\ b
  %   \end{array}
  %   $
  % };
\end{tikzpicture}

%%% Local Variables:
%%% mode: latex
%%% TeX-master: "../linearization_through_SOS"
%%% End:
  }
  \caption{\CIF\ model of a train.}
  \label{fig:train}
\end{figure}

These models can be composed in parallel using the parallel
composition operator, denoted as $\parallel$. Actions in \CIF\ are not
synchronizing by default. Thus in the parallel composition
\begin{equation*}
  \mathit{Train(0)} \parallel \mathit{Train(1)}
\end{equation*}
the actions of the two trains will be interleaved.

We want to put the parallel composition of the two trains in parallel
with the gate automaton, in such a way that the trains synchronize
with the actions $\mathit{rq}$, $\mathit{go}$, and $\mathit{out}$ of
the gate. This can be achieved using the \emph{synchronizing action}
operator, denoted as $\syncfs_A$. Informally, composition
$\sync{A}{p}$ behaves as composition $p$, except that all the actions
of the set $A$ are made synchronizing in $p$. Below we explain this.
Using these operators, we can express train gate model in \CIF\ as follows:
\begin{equation}
  \label{eq:2}
  \sync{\{\mathit{rq}, \mathit{go},
    \mathit{out}\}}{\mathit{Train(0)} \parallel
  \mathit{Train(1)}}\parallel \mathit{Gate}
\end{equation}
As a consequence of the use of the synchronizing action operator in
\eqref{eq:2}, action $i \in \{\mathit{rq}, \mathit{go},
\mathit{out}\}$ in $\mathit{Train(j)}$, $j \in \{0, 1\}$, will
synchronize with action $i$ in the gate. Actions in the set
$\{\mathit{rq}, \mathit{go}, \mathit{out}\}$ are interleaved in the
parallel compositions of the trains (they \emph{do not} synchronize)
since the scope operator only make actions synchronizing in the outer
scope. For more details see the rules of and their explanation
Table~\ref{tab:explicit-rules}.

% Note that in the above model the actions of the two trains are still
% interleaved, which in this case
% is the desired behavior.

Formally, the set of all \CIF\ compositions is defined as follows:
\begin{definition}[Compositions]
  The set $\compositions$ of all compositions is defined through the
  following abstract grammar: $\compositions ::= \alpha\ |\
  \compositions \parallel \compositions \ |\ \sync{A}{\compositions}$,
  where $\alpha$ is an automaton and $A \subseteq \actions$.
\end{definition}

In the next section we present the formal semantics of \CIF\ compositions,
both its explicit version and  its symbolic counterpart.

\subsection{Explicit and Symbolic Semantics of \CIF}
\label{sec:expl-symb-semant}

The semantics of \CIF\ is defined in terms of hybrid transition
systems~\cite{CuijpersReniers:LostInTranslation08}.
% \attention{What do these consist of?}
In the context of the present work, we restrict our attention to
ordinary transition systems (thus omitting time
  transitions), extended with \emph{environment transitions} (see
below).

The labeled transition systems we are considering have states of the
form $(p, \sigma)$. Here $p \in \compositions$, and $\sigma \in
\Sigma$ is a valuation, where $\Sigma = \variables \rightarrow
\values$, and $\values$ denotes a set of values. The valuation records
the values of the model variables at a certain moment.
There are two types of transitions in these labeled transition systems.
\emph{Action transitions}, of the form $$(p, \sigma) \trans{a, b} (p',
\sigma')$$ model the execution of an action $a$ by composition $p$ in
an initial valuation $\sigma$, which changes composition $p$ into $p'$
and results in a new valuation $\sigma'$.
Label $b$ is a boolean that indicates whether action $a$ is
synchronizing. \emph{Environment transitions}, of the form $$(p,
\sigma) \etrans{A} (p', \sigma')$$ model the fact that the initial
conditions and invariants of $p$ ($p'$ respectively) are satisfied in
$\sigma$ ($\sigma'$), and $A$ is the set of synchronizing
actions of $p$ and $p'$. Environment transitions are used to obtain
the state changes allowed by a model in a parallel composition
context.

The transition system associated to a composition can be obtained by
means of SOS rules. Below we present the explicit rules, where we have
omitted the symmetric version of the parallel composition rule. Given
a valuation $\sigma$, we define $\sigma'^+ \triangleq \{ (x^+, v) \ |\
(x, v) \in \sigma\}$. We use notation $\alpha$ to refer to the
automaton $ (V, \actv, \inv, E, \actS) $, and $\alpha[x]$ to refer to
$ (V, \ipred{x}, \inv, E, \actS)$, where $\ipred{x}(w) \triangleq w
\equiv x$. Throughout this work, $\mathrm{FV}(p)$ is the set of free
variables of $p$.

\begin{table}[htb]
  \centering
  \begin{tabular}{c c}
%    \hline \\
    \begin{minipage}{.4\linewidth}
      \Sosrule{
        (v, a, r, v') \in E,
        \sigma \models \actv(v) \wedge \inv(v),\\
        \sigma' \models \inv (v'),
        \sigma'^+ \cup \sigma \models r,\\
        \langle \forall x :: x^+ \notin \textrm{FV}(r) \Rightarrow \sigma
        (x) = \sigma'(x) \rangle
      }{
        (\alpha, \sigma)
        \trans{a, a \in \actS}
        (\alpha[v'], \sigma')
      }
      {rule:action:automata}
    \end{minipage} &
    \begin{minipage}{.4\linewidth}
      \Sosrule{
        v \in V,
        \sigma \models \actv(v) \wedge \inv(v),\\
        \sigma' \models \inv (v)
      }{
        (\alpha, \sigma)
        \etrans{\actS}
        (\alpha[v], \sigma')
      }
      {rule:environment:automata}
    \end{minipage}\\
    & \\
    \begin{minipage}{.4\linewidth}
      \Sosrule{
        (p, \sigma) \trans{a, \true} (p', \sigma'),
        (q, \sigma) \trans{a, \true} (q', \sigma')
      }{
        (p \parallel q, \sigma) \trans{a, \true} (p' \parallel q', \sigma')
      }{rule:action:sync-pc}
    \end{minipage} &
    \begin{minipage}{.45\linewidth}
      \Sosrule{
        (p, \sigma) \trans{a, b} (p', \sigma'),
        (q, \sigma) \etrans{A} (q', \sigma'), a \notin A
      }{
        (p \parallel q, \sigma) \trans{a, b}
        (p' \parallel q', \sigma')
      }{rule:action:interleaving-pc}
    \end{minipage}\\
    & \\
    \begin{minipage}{.45\linewidth}
      \Sosrule{
        (p, \sigma) \etrans{A_p} (p', \sigma'),
        (q, \sigma) \etrans{A_q} (q', \sigma')
      }{
        (p \parallel q, \sigma) \etrans{A_p \cup A_q} (p' \parallel q', \sigma')
      }{rule:environment:pc}
    \end{minipage} &
    \begin{minipage}[htb]{.4\linewidth}
      \Sosrule
      {
        (p , \sigma)
        \trans{a, b, X}
        (p', \sigma')
      }
      {
        (\sync{A}{p} , \sigma)
        \trans{a, b \vee a \in A, X}
        (\sync{A}{p'}, \sigma')
      }
      {rule:action:synchronization0}
    \end{minipage} \\ & \\
    \multicolumn{2}{c}{
      \begin{minipage}[htb]{0.4\linewidth}
        \Sosrule
        {
          (p,\sigma)
          \etrans{A'}
          (p', \sigma')
        }{
          (\sync{A}{p},\sigma)
          \etrans{A \cup A'}
          (\sync{A}{p'}, \sigma')
        }{rule:consistency:synchronization}
      \end{minipage}} \\ \\
%    \hline
  \end{tabular}
  \caption{Explicit rules for \CIF}
  \label{tab:explicit-rules}
\end{table}

Rule~\ref{rule:action:automata} states that an action can be triggered
by an automaton, if there is an edge $(v, a, r, v')$ such that the
initial predicate and the invariant are satisfied in the initial
valuation $\sigma$, and it is possible to find a new valuation
$\sigma'$ in which the invariant and the reset predicate are
satisfied. The only variables that change in $\sigma'$ w.r.t. $\sigma$
are those free variables of $r$ that are of the form $x^+$.
Rule~\ref{rule:environment:automata} states that an automaton is
consistent in initial valuation $\sigma$ if the initial predicate and
invariant are satisfied in $\sigma$, and the valuation can be changed
to $\sigma'$ only if the invariant is preserved.
% Explanation of the parallel composition semantics
Rule~\ref{rule:action:sync-pc} expresses that an action $a$ can be
executed synchronously if it is \emph{marked as synchronizing} in both
components. The interleaving behavior is modeled in
Rule~\ref{rule:action:interleaving-pc}, where an action $a$ can be
executed in $p$ if \emph{it is not synchronizing} in $q$.
% Explanation of the synchronizing action operator semantics
In Rule~\ref{rule:action:synchronization0} an action $a$ is marked as
synchronizing if $a \in A$, or $a$ is synchronizing in $p$. The
environment rule for the synchronizing action operator
(Rule~\ref{rule:consistency:synchronization}) adds $A$ to the set of
synchronizing actions of $p$.

As noted in~\cite{NadalesReniers:DerivingSimulator11}, the explicit
rules are not suitable for implementation purposes.
% Why? (From our SEFM2011 paper)
These rules often induce infinitely branching transition systems, and
as a consequence it is not possible to obtain the set of possible
successor states. In particular, the labels of the hybrid transition
systems contain \emph{trajectories}, of an dense domain, which are
defined in the rules through computations over these dense sets.
Another problem is that the valuations specify implicit constraints,
such as ``variables owned by a certain automaton cannot be changed in
a parallel composition'', which require to compute operations on
infinite sets of valuations to get the set of possible successor
states.

The solution to the problem explained above was to obtain a set of
\emph{symbolic rules}~\cite{Hennessy:SymbolicBisimulations1995} from
the explicit SOS specification. These
symbolic rules represent the possible state changes by means of
predicates, and thus, the state change caused by an action is visible
on the arrows of the transitions. The symbolic rules for the language
considered in this paper are shown in Table~\ref{tab:symbolic-rules}.

\begin{table}[htb]
  \centering
  \begin{tabular}{c c}
%    \hline \\&\\
    \begin{minipage}{.40\linewidth}
      \Sosrule{
        (v, a, r, v') \in E
      }{
        \sstate{\alpha}
        \trans{a, a \in \actS, \actv(v), \inv(v), \inv(v'), r}
        \sstate{\alpha[v']}
      }
      {srule:action:automata}
    \end{minipage} &
    \begin{minipage}{.40\linewidth}
      \Sosrule{
        v \in V
      }{
        \sstate{\alpha}
        \etrans{\actv(v), \inv(v),\actS}
        \sstate{\alpha[v]}
      }
      {srule:environment:automata}
    \end{minipage}\\
    & \\
%    \multicolumn{2}{c}{
      \begin{minipage}{0.45\linewidth}
        \Sosrule{
          \sstate{p} \trans{a, \true, u_p, n_p, n_p', r_p} \sstate{p'},
          \sstate{q} \trans{a, \true, u_q , n_q, n_q', r_q} \sstate{q'}
        }{
          \sstate{p \parallel q}
          \trans{a, \true, u_p \wedge u_q, n_p \wedge
            n_q, n_p' \wedge n_q', r_p \wedge r_p'}
          \sstate{p' \parallel q'}
        }{srule:action:sync-pc}
      \end{minipage} &
%    } \\ & \\
%    \multicolumn{2}{c}{
      \begin{minipage}{.450\linewidth}
        \Sosrule{
          \sstate{p} \trans{a, b, u_p, n_p, n_p', r} \sstate{p'},
          \sstate{q} \etrans{u_q, n_q, A} \sstate{q'}, a \notin A
        }{
          \sstate{p \parallel q} \trans{a, b, u_p \wedge u_q, n_p
            \wedge n_q, n_p' \wedge n_q, r}
          \sstate{p' \parallel q'}
        }{srule:action:interleaving-pc}
      \end{minipage} \\ & \\
%    }
      \begin{minipage}{0.45\linewidth}
        \Sosrule{
          \sstate{p}
          \setrans{u_p, n_p, A_p}
          \sstate{p'},
          \sstate{q}
          \setrans{u_q, n_q, A_q}
          \sstate{q'}
        }{
          \sstate{p \parallel q}
          \setrans{u_p \wedge u_q, n_p \wedge n_q, A_p \cup A_q}
          \sstate{p' \parallel q'}
        }{rule:environment:pc}
      \end{minipage} &
      \begin{minipage}{0.45\linewidth}
        \Sosrule{
          \sstate{p}
          \satrans{a, b, u, n, n', r}
          \sstate{p'}
        }{
          \sstate{\sync{A}{p}}
          \satrans{a, b \vee a \in A, u, n, n', r}
          \sstate{\sync{A}{p'}}
        }{srule:action:sync}
      \end{minipage}\\ & \\
      \multicolumn{2}{c}{
        \begin{minipage}{1\linewidth}
          \Sosrule{
            \sstate{p}\setrans{u,n,A'}\sstate{p'}
          }{
            \sstate{\sync{A}{p}}
            \setrans{u,n,A \cup A'}
            \sstate{\sync{A}{p'}}
          }{srule:environment:sync}
        \end{minipage}
      }\\ & \\
%      \hline
  \end{tabular}
  \caption{Symbolic rules for \CIF}
  \label{tab:symbolic-rules}
\end{table}

The explicit and symbolic rules are related by the following soundness
and completeness theorems.
These theorems state how an explicit transition system can be
reconstructed from its symbolic version, and vice-versa. % We can see
% these results as a kind of \emph{refinement relation}.

\begin{theorem}[Soundness of action transitions]
  \label{theo:soundness-atrans}
  For all $p$, $p'$, $a$, $b$, $u$, $n$, $n'$, $r$,
  $\sigma$, and $\sigma'$ we have that if the following conditions
  hold:
  \begin{enumerate}
  \item $\sstate{p} \satrans{a, b, u, n, n', r} \sstate{p'}$
  \item
   $\sigma \models u$,
   $\sigma \models n$,
   $\sigma' \models n'$, and  $\sigma'^+ \cup \sigma \models r$
 \item $ \langle \forall x :: x^+ \notin \textrm{FV}(r) \Rightarrow
   \sigma (x) = \sigma'(x) \rangle$
  \end{enumerate}
  then, there is a explicit action transition
  $(p, \sigma) \trans{a, b} (p', \sigma')$.
\end{theorem}

\begin{theorem}[Completeness of action transitions]
  \label{theo:completeness-atrans}
  For all $p$, $p'$, $a$, $b$, $\sigma$, and $\sigma'$ we have
  that if there is a explicit transition
  $(p, \sigma) \trans{a, b} (p', \sigma')$
  then there exists $u$, $n$, $n'$, and $r$ such that the following
  conditions hold:
  \begin{enumerate}
  \item $\sstate{p} \satrans{a,
      b, u, n, n', r} \sstate{p'}$
  \item $\sigma \models u$, $\sigma \models n$,
    $\sigma' \models n'$, and $\sigma'^+ \cup \sigma \models r$
  \item $\langle \forall x :: x^+ \notin \textrm{FV}(r) \Rightarrow
   \sigma (x) = \sigma'(x) \rangle$
  \end{enumerate}
\end{theorem}

\begin{theorem}[Soundness of environment transitions]
  \label{theo:soundness-etrans}
  For all $p$, $p'$, $u$, $A$,
  $\sigma$, and $\sigma'$ we have that if the following conditions
  hold:
  \begin{enumerate}
  \item $\sstate{p} \setrans{u, n, A} \sstate{p'}$
  \item $\sigma \models u$, $\sigma \models n$, $\sigma' \models n$
  \end{enumerate}
  then, there is a explicit environment transition
  $(p, \sigma) \etrans{A} (p', \sigma')$.
\end{theorem}

\begin{theorem}[Completeness of environment transitions]
  \label{theo:completeness-etrans}
  For all $p$, $p'$, $A$, $\sigma$, and $\sigma'$ we have that if
  there is an explicit transition $(p, \sigma) \etrans{A} (p',
  \sigma')$ then there exists $u$, and $n$ such that the following
  conditions hold:
  \begin{enumerate}
  \item $\sstate{p} \setrans{u, n, A} \sstate{p'}$
  \item $\sigma \models u$, $\sigma \models n$, $\sigma' \models n$
  \end{enumerate}
\end{theorem}

It is not hard to see that given a \CIF\ composition, the symbolic
rules induce a \emph{finite transition system}. For the model of the
train gate presented in Section~\ref{sec:setting-scene}, a part of its
associated symbolic transition system is shown in
Figure~\ref{fig:sts-gateway-controller} (the whole transition system
contains 16 states), where we use the convention that for all $x$,
$y$, $z$:
\begin{equation*}
  \sstate{x, y, z} \equiv \sync{\{\mathit{rq}, \mathit{go},
    \mathit{out}\}}{\mathit{Train(0)}[x] \parallel
    \mathit{Train(1)}[y]} \parallel \mathit{Gate}[z]
\end{equation*}
In this transition system, two problems can be noted. The size of the
symbolic transition system grows exponentially as more trains are
added. This is the result of the interleaving actions that are
executed between these models.
Secondly, there is a great deal of \emph{redundant information}. The
invariants of the source and the target states are present not only in
the labels of action transitions, but also in the environment
transition of these
states. Similarly, the initialization conditions are meaningful only
for the initial environment transition. For the remaining environment
transitions in the systems, the initialization predicate is always
true. In the next section we show how to overcome these
problems using a new kind of symbolic rules.

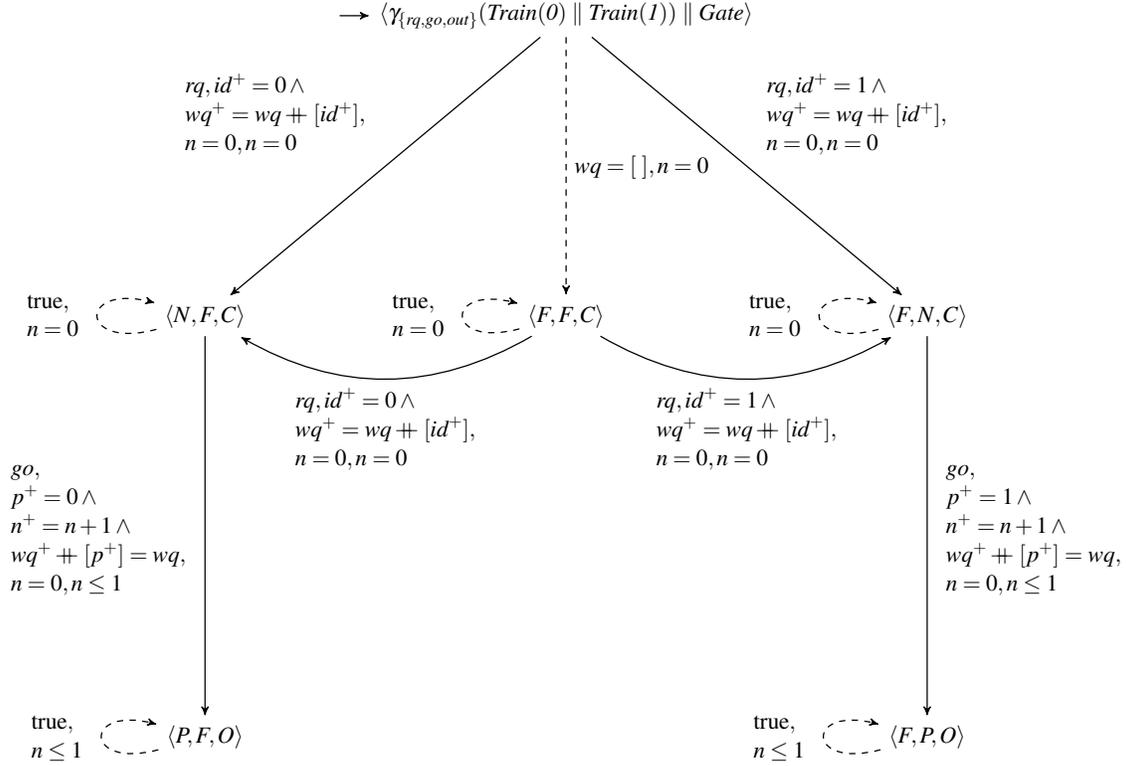
\begin{figure}[htb]
  \centering
  \scalebox{0.8}{
    \begin{tikzpicture}[->,>=stealth',shorten >=1pt,auto, semithick,
  initial text=, node distance=13em]

  \node[initial] (start) { $\sstate{ \sync{\{\mathit{rq}, \mathit{go},
      \mathit{out}\}}{\mathit{Train(0)} \parallel
      \mathit{Train(1)}}\parallel \mathit{Gate}}
    $ };

  \node (1) at ($(start.center)+(-6, -5)$) {$\sstate{N, F, C}$};

  \node (2) at ($(start.center)+(0, -5)$) {$\sstate{F, F, C}$};

  \node (3) at ($(start.center)+(6, -5)$) {$\sstate{F, N, C}$};

  \node (4) at ($(start.center)+(-6, -12)$) {$\sstate{P, F, O}$};

%  \node (5) at ($(start.center)+(0, -12)$) {$\sstate{N, N, C}$};

  \node (6) at ($(start.center)+(6, -12)$) {$\sstate{F, P, O}$};

  \path (start) edge [dashed] node {$\mathit{wq}=[\ ], n = 0$} (2);

  \path (start) edge node [above left] {$
    \begin{array}{l}
      \mathit{rq},
      \mathit{id}^+=0 \wedge{}\\
      \mathit{wq}^+=\mathit{wq} \concat [\mathit{id}^+],\\
      n=0, n=0
    \end{array}
    $} (1);

  \path (start) edge node [above right] {$
    \begin{array}{l}
      \mathit{rq},
      \mathit{id}^+=1 \wedge{}\\
      \mathit{wq}^+=\mathit{wq} \concat [\mathit{id}^+],\\
      n=0, n=0
    \end{array}
    $} (3);

  \path (1) edge [dashed,loop left] node {$
    \begin{array}{l}
      \true, \\ n=0
    \end{array}
    $} (1);

  \path (4) edge [dashed,loop left] node {$
    \begin{array}{l}
      \true, \\ n \leq 1
    \end{array}
    $} (4);

  \path (6) edge [dashed,loop left] node {$
    \begin{array}{l}
      \true, \\ n \leq 1
    \end{array}
    $} (6);

  \path (2) edge [dashed,loop left] node {$ \begin{array}{l} \true, \\
      n=0
    \end{array}$} (2);

  \path (3) edge [dashed,loop left] node {$
    \begin{array}{l}
      \true, \\ n=0
    \end{array}
    $} (3);

  \path (2) edge [bend left] node [below] {$
    \begin{array}{l}
      \mathit{rq},
      \mathit{id}^+=0 \wedge{}\\
      \mathit{wq}^+=\mathit{wq} \concat [\mathit{id}^+],\\
      n=0, n=0
    \end{array}
    $} (1);

  \path (2) edge [bend right] node [below] {$
    \begin{array}{l}
      \mathit{rq},
      \mathit{id}^+=1 \wedge{}\\
      \mathit{wq}^+=\mathit{wq} \concat [\mathit{id}^+],\\
      n=0, n=0
    \end{array}
    $} (3);

  \path (1) edge node [left] {
    $
    \begin{array}{l}
      \mathit{go},\\
      p^+=0 \wedge{} \\ n^+=n+1 \wedge{} \\
      \mathit{wq}^+ \concat [p^+] = \mathit{wq},\\
      n = 0, n \leq 1
    \end{array}
    $
  } (4);

  \path (3) edge node [right] {
    $
    \begin{array}{l}
      \mathit{go},\\
      p^+=1 \wedge{} \\ n^+=n+1 \wedge{} \\
      \mathit{wq}^+ \concat [p^+] = \mathit{wq},\\
      n = 0, n \leq 1
    \end{array}
    $
  } (6);
\end{tikzpicture}
%%% Local Variables:
%%% mode: latex
%%% TeX-master: "../linearization_through_SOS"
%%% End:
  }
  \caption{A part of the symbolic transition system for the train gate
    controller}
  \label{fig:sts-gateway-controller}
\end{figure}

\section{Linear Transition Systems}
\label{sec:linear-transition-systems}

In this section we define a structure called \emph{linear transition
  system} (LiTS),
which contains all the information necessary to represent
any arbitrary \CIF\ composition, and that can be
translated to an equivalent automaton.

%%% About the states of the LiTS
Consider the symbolic transition system of the train gate model. In
Figure~\ref{fig:sts-gateway-controller}, we show a transition of the
form:
\begin{equation*}
  \sstate{N, F, C} \trans{\mathit{go},
      p^+=0 \wedge n^+=n+1 \wedge
      \mathit{wq}^+ \concat [p^+] = \mathit{wq},
      n = 0, n \leq 1} \sstate{P, F, O}
\end{equation*}
The complete symbolic transition system also contains these
transitions:
\begin{flalign*}
  \sstate{N, N, C} \trans{\mathit{go},
    p^+=0 \wedge n^+=n+1 \wedge
    \mathit{wq}^+ \concat [p^+] = \mathit{wq},
    n = 0, n \leq 1} \sstate{P, N, O}\\
  \sstate{N, S, C} \trans{\mathit{go},
    p^+=0 \wedge n^+=n+1 \wedge
    \mathit{wq}^+ \concat [p^+] = \mathit{wq},
    n = 0, n \leq 1} \sstate{P, S, O}
\end{flalign*}
These three transitions only differ in the second component of the
symbolic state, that is, the location in which the second train is.
However, this information is not relevant for computing the state
change. If we replace the above transitions by a unique transition of
the form:
\begin{equation*}
  \sstate{N, \_, C} \trans{\mathit{go},
      p^+=0 \wedge n^+=n+1 \wedge
      \mathit{wq}^+ \concat [p^+] = \mathit{wq},
      n = 0, n \leq 1} \sstate{P, \_, O}
\end{equation*}
then we can avoid the state explosion caused by the interleaving
actions. Here the \emph{wild-card} symbol $\_$ can be read as ``for
any location''.

Furthermore, in Figure~\ref{fig:sts-gateway-controller} we see that
there is no need to replicate the entire structure in a given
transition, since it suffices to keep track of the locations that
change.

From the observation above, we want a linear transition system where
the states are \emph{sequences of locations}, containing also
\emph{wild-cards}. These wild-cards are used to denote the fact that
the location of a certain automaton does not change in the transition.
Formally the states of the LiTS belong to the set
\begin{equation}
  \label{eq:1}
  (\locations \times \{\_\})^*
\end{equation}
where $\_$ is the wild-card symbol, and $A^*$ is the set of all
sequences whose elements are taken from the set $A$. An example of
such state is the list $[F, \_, C]$.

The next thing to define is the transitions of the LiTS's, in such a
way that the redundancy introduced by the STS's is eliminated. To
accomplish this, we split action and environment transitions into
several transitions, which are described next.

\begin{description}
\item[Action Transitions] They are of the form
$
  p \models \sstate{ \xs v} \trans{a, r} \sstate{\xs v'}
$,
where $p \in \compositions$ is a composition, $a \in \actionstau$ is
an action label, and $r \in \predicates$ is the update predicate
associated to the action.
\item[Synchronizing Actions] They are of the form
$
  p \saof A
$,
where $p \in \compositions$ is a composition, and $A \subseteq
\actions$ is the set of synchronizing actions of $p$.
\item[Initialization Transitions] They are of the form $p \ipredof \xs
  f$, where $p \in \compositions$ and $\xs f \in (\locations
  \rightharpoonup \predicates)^*$ is a list containing the
  initialization predicate function of each automaton in $p$.
\item[Invariant Transitions] They are of the form $ p \invof \xs f $,
  where $p \in \compositions$ is a composition, and $\xs f \in
  (\locations \rightarrow \predicates)^*$ is a list containing the
  invariant function associated to each automaton in $p$.

  The reader may have expected initialization or invariant
    transitions of the form:
  \begin{equation*}
    \xs v \invof p
  \end{equation*}
  where $\xs v$ is a list of locations, and $p$ is a predicate.
  However this approach requires enumerating the state space
  explicitly to construct the $\invof$ relation. By using lists of
  functions we avoid this explicit construction.
\item[Wild-card Transitions] They are of the form $ p \compsof \xs x $,
  where $p \in \compositions$ is a composition, and $\xs x \in (\_)^*$
  is a sequence of wild-cards whose size coincides with the number of
  automata that are composed in parallel in $p$. These transition are
  not needed for reconstructing the environment transitions, they are
  used in the linear SOS rules to model the fact that nothing changes
  in a component of a parallel composition, when the other component
  performs an action.
\end{description}

% A requirement of the initialization and invariant transitions is that
% they have to be \attention{linear} in size (with regard to $p$). Thus,
% we use lists because if we used conjunctions, we would have had to
% enumerate all the states explicitly, which is exactly what we are
% trying to avoid. \attention{I don't understand this!!!!}

In Table~\ref{tab:linear-rules} we show some of the linear SOS rules for
\CIF\ compositions. We have omitted the rules for synchronizing
actions, initialization, and wild-card transitions since they are
similar to the invariant transitions.

% \attention{Explanation of some rules is in order.}
The linear rules can be easily to obtained from the symbolic
  ones. For action rules, invariants and initialization predicates,
and the synchronizing action label are simply omitted (since they can
be obtained from other transitions). The linear rule for
interleaving parallel composition is almost identical to the symbolic
rule. The only differences are that the set $A$ is obtained from a
$\saof$ transition, and we use the wild-card transition to represent
the fact that the locations of the other automaton are not relevant
(at the symbolic level at least).
A similar observation can be made for the rule for parallel
composition. In this case since we do not have the synchronizing
label, we reconstruct it from the $\saof$ transition. This label is
equivalent to $a \in A$, thus a label $\true$ in both components
is equivalent to $a \in A_p \wedge a \in A_q$, which is in turn equivalent
to $a \in A_p \cap A_q$.

\begin{table}[htb]
  \centering
  \begin{tabular}{c c}
%    \hline & \\
    \begin{minipage}{.40\linewidth}
      \Sosrule{}
      {(V, \actv, \inv, \tcp, E, \actS)\invof  [\inv]}
      {lrule:inv:atomicaut}
    \end{minipage}&
    \begin{minipage}{.40\linewidth}
      \Sosrule{p \invof \xs f_p, q \invof \xs f_q}
      {p \parallel q \invof \xs f_p \concat \xs f_q}
      {lrule:inv:pcomp}
    \end{minipage}\\ & \\
    \begin{minipage}{.40\linewidth}
      \Sosrule{
        (v, a, r, v') \in E
      }
      {(V, \actv, \inv, \tcp, E, \actS)\models
        \sstate{[v]}\trans{a,r}\sstate{[v']}}
      {lrule:action:atomicaut}
    \end{minipage} &
    \begin{minipage}{.40\linewidth}
      \Sosrule{
        p \models \sstate{\xs v}\trans{a, r}\sstate{\xs v'},
        q \saof A,
        q \compsof \_,
        a \notin A,
      }{
        p \parallel q \models \sstate{\xs v \concat \_}
        \trans{a, r}\sstate{\xs v' \concat \_ }
      }
      {lrule:action:interleaving-pc}
    \end{minipage}\\ & \\
    \multicolumn{2}{c}{
      \begin{minipage}{0.9\linewidth}
        \Sosrule{
          p \models \sstate{\xs v_p} \trans{a, r_p} \sstate{\xs v_p'},
          q \models \sstate{\xs v_q} \trans{a, r_q} \sstate{\xs v_q'},
          p \saof A_p, q \saof A_q, a \in A_p \cap A_q
        }{
          p \parallel q \models
          \sstate{\xs v_p \concat \xs v_q}
          \trans{a, r_p \wedge r_q}
          \sstate{\xs v_p' \concat \xs v_q'}
        }{lrule:action:sync-pc}
      \end{minipage}
    }\\ & \\
    \begin{minipage}{0.45\linewidth}
      \Sosrule{
        p \invof \xs f
      }{
        \sync{A}{p} \invof \xs f
      }{lrule:inv:sync}
    \end{minipage}&
    \begin{minipage}{0.45\linewidth}
      \Sosrule{
        p \models \sstate{\xs v} \trans{a, r} \sstate{\xs v'}
      }{
        \sync{A}{p} \models \sstate{\xs v} \trans{a, r} \sstate{\xs v'}
      }{lrule:action:sync}
    \end{minipage}
    \\ & \\
%    \hline
  \end{tabular}
  \caption{Linear SOS rules for \CIF\ compositions}
  \label{tab:linear-rules}
\end{table}

If a composition $p$ contains no synchronizing actions, then the size
of its induced transition system is linear w.r.t. the size of $p$.
However, the size of the LiTS also depends on the number of
synchronizing actions. The following property gives the formal
details.

\begin{prop}[Size of the linear transition system]
  Let $p$ be a \CIF\ composition, such that it contains $n$ automata
  $\alpha_i \equiv (V_i, \actv_i, \inv_i, E, {\actS}_i)$, $0 \leq i <
  n$. Let $a$ be the only synchronizing action in these automata.
  Then the number of transitions in the LiTS associated to $p$ is
  given by:
  \begin{equation}
    \label{eq:11}
    \displaystyle\sum_{0 \leq i < n} \#\{ x\ |\ (v, g, x, u, v' ) \in
    E_i \wedge x \neq a\} + \displaystyle\prod_{0 \leq i < n} \#\{ x\ |\ (v,
    g, x, u, v' ) \in E_i \wedge x = a\}
  \end{equation}
  where $\#A$ is the number of elements in set $A$.
\end{prop}
% TODO: More specifically ...

In spite of the fact that the number in~\eqref{eq:11} can be
significantly large, in practice, communication among components is
usually restricted to a few automata, and the number of edges of an
automaton that contain a given synchronizing action $a$ is small.

% Let $p$ be a composition containing $n$ automata, such that they
% synchronize on actions of a given set $A$. Then the number of states
% of the LiTS
% induced by $p$ size is
% The size of the linear transition system depends on the synchronizing
% \todo{We should say something about the size of the resulting LiTS.}

\subsection{Relating LiTS and STS}
\label{sec:relat-lits-sts}

In the same way symbolic transitions are related to explicit ones via
soundness and completeness results, linear transitions have the same
property w.r.t. symbolic transitions.%  These results can be seen as the \emph{refinement relation}
% among the different SOS specifications.
% From the practical point of
% view, soundness and completeness results are useful for obtaining

The first two results state that a LiTS contains all the necessary
information to reconstruct the environment transitions in the
symbolic transition system and vice-versa. Here ``leads
to transitions'' refers to the
initialization, invariant, and synchronizing actions transitions in
the LiTS. Given a composition $p$, which contains $n$ atomic automata,
and a sequence $\xs l$ of $n$ locations, $p[\xs l]$ is the composition
obtained by replacing the initial predicate function of the $i^{th}$
automaton by $\ipred{\xs l {.} i}$, for $0 \leq i < n$, where $\xs l {.} i$ is
the element of sequence $\xs l$ at position $i$ (sequences are
numbered starting from $0$). $\mathrm{locsof}(p)$ refers to the set of
sequences $\xs l$, where $\# \xs l = n$ and $\xs l {.} i$ is a location
of the $i^{th}$ automaton of composition $p$ ($0 \leq i < n$).

\begin{theorem}[Soundness of leads to transitions]
  For all $p$, $\xs i$, $\xs f$, $\xs g$, $A$, $u$, and $n$ we have
  that if the following conditions hold:
  \begin{enumerate}
  \item $\xs i \in \mathrm{locsof}(p)$
  \item $p \ipredof \xs f$, $p \invof \xs g$, $p \saof A$
  \item $u = \displaystyle \bigwedge_{0 \leq i < \#\xs f} \xs f{.}i
    (\xs i{.}i)$, and $n = \displaystyle \bigwedge_{0 \leq i < \# \xs
      g} \xs g {.} i (\xs i {.} i)$
  \end{enumerate}
  then there is a symbolic transition
$
    \sstate{p} \etrans{u, n, A} \sstate{p[\xs i]}
$.
\end{theorem}

% \begin{proof}
%   Relatively easy. Proof done on paper.
% \end{proof}

\begin{theorem}[Completeness of leads to Transitions]
  For all $p$, $u$, $n$, $A$, and $p'$ we have that if there is an
  environment transition $\sstate{p} \etrans{u, n, A} \sstate{p'}$
  then there are $\xs i$, $\xs f$, $\xs g$, $u$, and $n$ such that the
  following conditions hold:
  \begin{enumerate}
  \item $\xs i \in \mathrm{locsof}(p)$
  \item $p \ipredof \xs f$, $p \invof \xs g$, $p \saof A$
  \item  $u = \displaystyle \bigwedge_{0 \leq i < \#\xs f} \xs f{.}i
    (\xs i{.}i)$, and $n = \displaystyle \bigwedge_{0 \leq i < \# \xs
      g} \xs g {.} i (\xs i {.} i)$,  $p' \equiv p[\xs i]$
  \end{enumerate}
\end{theorem}

% Condition $p' \equiv p[\xs i]$ in the theorem above is needed for
% showing completeness of linear action transitions. \attention{Abracadabra}
% \begin{proof}
%   Relatively easy. Proof done on paper.
% \end{proof}

The soundness theorem for linear action transitions shows how a
symbolic action transition can be obtained, using the leads to
transitions as well. Functions $\sqsubseteq$ and $\overwrite$ are
defined below, where $x:\xs x$ is
the list that results after appending the element $x$ to the front of $\xs x$.

\begin{definition}[Sub-sequence and sequence overwriting] Function
  $\sqsubseteq \in A^* \rightharpoonup A^* \rightharpoonup \setofbooleans$
  is defined as follows:
  \begin{flalign*}
    [\ ] \sqsubseteq \xs x & \triangleq \true\\
    (x: \xs x) \sqsubseteq (y:\xs y) & \triangleq ((x \equiv \_) \vee
    (x \equiv y)) \wedge \xs x \sqsubseteq \xs y
  \end{flalign*}
  Function $\overwrite \in A^* \rightharpoonup A^* \rightharpoonup
  A^*$ is defined as follows:
  \begin{flalign*}
    [\ ] \overwrite \xs x &  \triangleq \xs x\\
    (x:\xs x) \overwrite (y:\xs y) & \triangleq
    \begin{cases}
      x:(\xs x \overwrite \xs y) & \text{if } x \neq \_\\
      y:(\xs x \overwrite \xs y) & \text{if } x = \_
    \end{cases}
  \end{flalign*}
\end{definition}

\begin{theorem}[Soundness of Linear Action Transitions]\label{theo:soundness-lits}
  For all $p$, $\xs v$, $a$, $r$, $\xs v'$, $\xs i$, $\xs f$, $\xs g$,
  $u$, $n$, $n$,
  and $A$ we have that if the following conditions hold:
  \begin{enumerate}
  \item $\xs i \in \mathrm{locsof}(p)$
  \item $p \models \sstate{\xs v} \trans{a,r} \sstate{\xs v'}$, $p
    \ipredof \xs f$, $p \invof \xs g$, $p \saof A$
  \item $u = \displaystyle \bigwedge_{0 \leq i < \#\xs f} \xs f{.}i
    (\xs i{.}i)$, and $n = \displaystyle \bigwedge_{0 \leq i < \# \xs
      g} \xs g {.} i (\xs i {.} i)$, $n' =
    \displaystyle \bigwedge_{0 \leq i < \# \xs g} \xs g{.}i ((\xs v'
    \overwrite \xs i){.}i)$, $\xs v \sqsubseteq \xs i$, $b \equiv a
  \in A$
 \end{enumerate}
 then there is a symbolic transition: $
 \sstate{p}\trans{a,b,u,n,n',r}\sstate{p[\xs v' \overwrite \xs i]}$.
 %%% !!! It seems we DO NOT need the following part: !!! %%%
 % \text{ and }
 % \sstate{p[\xs v \overwrite \xs i]}
 % \trans{a,u,n,n',r}\sstate{p[\xs v' \overwrite \xs i]}
\end{theorem}

% \begin{proof}
%   Easy. Proof done on paper.
% \end{proof}

\begin{theorem}[Completeness of Linear Action Transitions]
  For all $p$, $p'$, $a$, $b$, $u$, $n$, $n'$, and $r$ we have that if
  there is a symbolic transition:
  \begin{equation*}
    \sstate{p} \satrans{a, b, u, n, n', r} \sstate{p'}
  \end{equation*}
  then there are $\xs v$, $\xs v'$, $\xs i$, $\xs f$, $\xs g$, and $A$
  such that the following conditions hold:
  \begin{enumerate}
  \item $\xs i \in \mathrm{locsof}(p)$
  \item $p \models \sstate{\xs v} \trans{a,r} \sstate{\xs v'}$, $p
    \ipredof \xs f$, $p \invof \xs g$, $p \saof A$
  \item $u = \displaystyle \bigwedge_{0 \leq i < \#\xs f} \xs f{.}i
    (\xs i{.}i)$, and $n = \displaystyle \bigwedge_{0 \leq i < \# \xs
      g} \xs g {.} i (\xs i {.} i)$, $n' =
    \displaystyle \bigwedge_{0 \leq i < \# \xs g} \xs g{.}i ((\xs v'
    \overwrite \xs i){.}i)$, $\xs v \sqsubseteq \xs i$, $b \equiv a
  \in A$, $p' \equiv p[\xs v' \overwrite \xs i]$
  \end{enumerate}
\end{theorem}

These theorems can be proved using structural induction.
  The proofs are relatively simple, and are omitted due to space
  constraints.
% simple because the linear SOS rules were obtained \attention{were
% obtained is to vague} from the symbolic ones.

% \begin{proof}
%   Relatively easy. Proof done on paper.
% \end{proof}

\section{Obtaining a Linear Automaton from a LiTS}
\label{sec:obta-line-autom}

Once a linear transition system is induced by the SOS rules, we need a
way to obtain a linear automaton from it. In this section we describe
the procedure, and we show that the generated automaton is stateless
bisimilar~\cite{MousaviRenGro:CongrSOSdataArtInfComp05} to the composition that induced the transition system. Both
from a theoretical and a practical point of view this is an
interesting result, which tells us that every composition can be
reduced to an automaton (this is intuitively obvious for the language
we present here, but it is not for \CIF\ and its hierarchical
extension).
% \attention{Besides an automaton also a state is produced! Shouldn't
% this be explained?} I don't think we should bother the reader with
% this detail now. But I'll explain the second component later.

Formally, given an composition $p$ and its associated LiTS $M$, we want to
build an automaton $\alpha_p$ such that $p$ has the same behavior as
$\alpha_p$.
The idea is to simulate the execution of $M$, using $\alpha_p$.
To this end, we need to introduce a sequence of variables $\xs l$,
which are used to represent the active state in $M$ in a given
execution. We call these variables \emph{location
  pointers}~\cite{Khadim3:LinHChiPCTech07}.
Below, we give the definition\footnote{Strictly speaking, function
  $\lfunfs$ is not uniquely determined, since it is possible to pick
  different location pointers. This can be avoided by defining a function
  that returns the least $n$ fresh variables in a given composition
  (assuming variables are totally ordered). A similar observation can
  be done about location $x$.} of the linearization
function, which returns the automaton associated to a given
composition and the location pointers used in it. The second component
returned by the function is used later to formulate the correctness result.

\begin{definition}[Linearization Function]
  Let $p$ be a \CIF\ composition. Function $\lfunfs \in \compositions
  \rightarrow (\compositions \times \variables^*)$ is defined as the
  least function that satisfies:
  \begin{equation*}
    \lfun{p} = ((\{x\}, \actv, \inv, E, \actS), \xs l)
  \end{equation*}
  where
  \begin{itemize}
  \item $p \compsof \xs x$, $\#\xs x = n$, $p \ipredof \xs f$, $p
    \invof \xs g$, $p \saof A$
  \item $\langle \forall i :: 0 \leq i < n \Rightarrow \xs l {.} i \notin
    \mathrm{FV}(p) \rangle$, $x \in \locations$
  \item $\actv(x) = (\displaystyle \bigwedge_{0 \leq i < n}\
    \displaystyle \bigwedge_{v \in \dom (\xs f{.}i)} (\xs l {.} i = v
    \Rightarrow \xs f {.} i(v)) ) \wedge (\displaystyle \bigwedge_{0
      \leq i < n} \xs l {.} i \in \dom(fs.i))$
  \item $\inv(x) = \displaystyle \bigwedge_{0 \leq i < n}\
    \displaystyle \bigwedge_{v \in \dom (\xs g{.}i)} (\xs l {.} i = v
    \Rightarrow \xs g {.} i(v))$
  \item $E = \{ (x, a, r \wedge \displaystyle \bigwedge_{
      \begin{array}{l}
         0 \leq i < n\\
        vs.i \neq \_
      \end{array}} \xs l {.} i = vs{.}i \wedge \xs l {.} i^+ =
    vs'{.}i, x) \ |\ p \models
    \sstate{\xs v}\trans{a, r}\sstate{\xs v'}\}$
  \end{itemize}
\end{definition}

In the above definition we introduce $n$ free variables, which are
used as location pointers, and we use a location $x$ (which can be
defined as the least location in $\locations$) as the unique location
of the automaton.
The initial predicate and invariant functions are conditional
expressions, which ensure that the right predicate is chosen according
to the values of the location pointers. In the definition of the
$\actv$ function, the second part of the conjunction forces the choice
of an initial location (otherwise this predicate can be trivially
satisfied).
The set of edges is constructed from the action transition of the
linear transition system. The reset mapping in the action transitions
is extended with updates to the location pointers to keep track of the
state in the linear transition system.

The well-definedness of function $\lfunfs$ is a consequence of the
finiteness of LiTSs. %\attention{Do we have this finiteness result?} I
                     %think it is obvious
%
Given a composition $p$, such that $\lfun{p} = (\alpha_p, \xs l)$, we
say that $\alpha_{p}$ is the linear automaton associated to it.

For the train gate model, the linear automaton associated to it is
shown in Figure~\ref{fig:linear-automaton}, where the initial
predicate and invariant functions are (once they are \emph{simplified}):
\begin{flalign*}
  \actv(x) & = ( l_0 = F \wedge l_1 = F \wedge l_2 = C \wedge
  \mathit{wq}
  = [\ ])\\
  \inv(x) & =  (l_2 = C \Rightarrow n = 0) \wedge (l_2 = O
  \Rightarrow n \leq 1)
\end{flalign*}

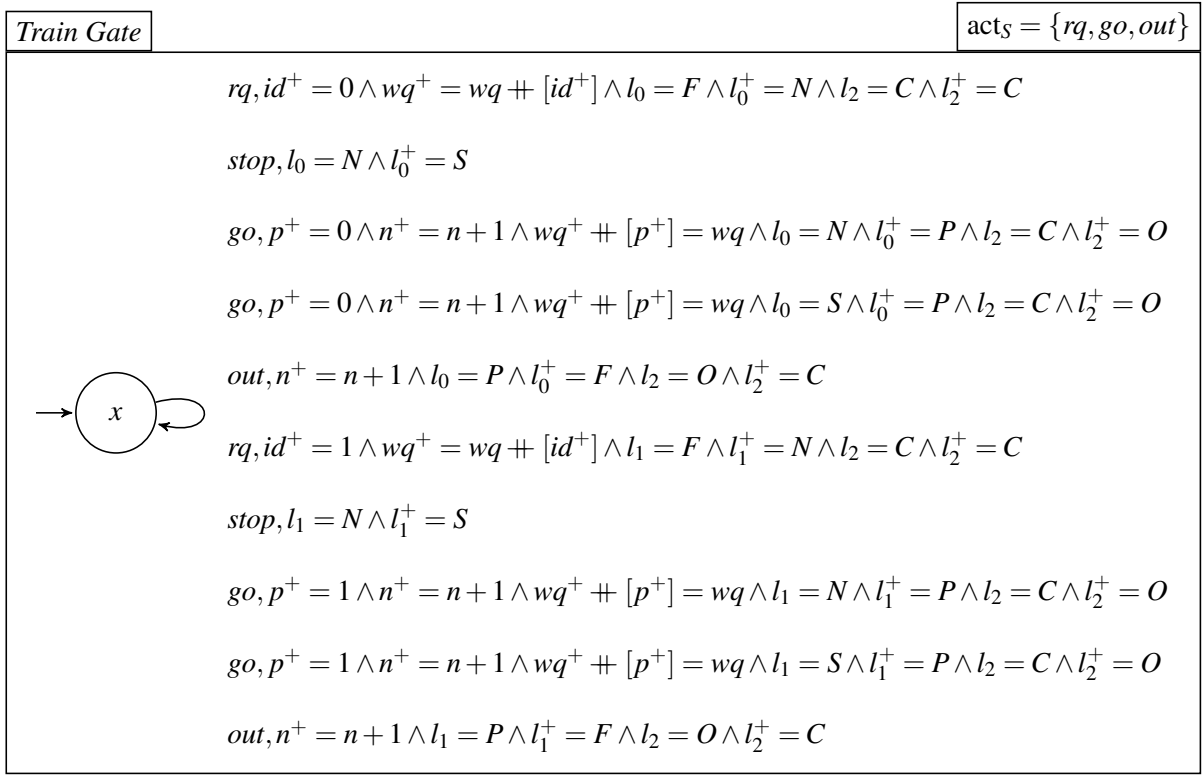
\begin{figure}[htb]
  \centering
  \scalebox{1}{
    \begin{tikzpicture}[->,>=stealth',shorten >=1pt,auto, semithick,
  initial text=, node distance=13em]

  \node[initial, state] (x) {
    $
    \begin{array}{c}
      x
    \end{array}
    $
  };

  \path (x) edge [loop right] node [right] {
    $
    \begin{array}{l}
      \mathit{rq}, \mathit{id}^+=0 \wedge
      \mathit{wq}^+=\mathit{wq} \concat [\mathit{id}^+] \wedge
      l_0=F \wedge l_0^+=N \wedge l_2=C \wedge l_2^+=C\\ \\
      \mathit{stop},
      l_0=N \wedge l_0^+=S\\ \\
      \mathit{go},  p^+=0 \wedge  n^+=n+1 \wedge
      \mathit{wq}^+ \concat [p^+] = \mathit{wq} \wedge l_0=N \wedge
      l_0^+=P \wedge l_2=C \wedge l_2^+=O\\ \\
      \mathit{go},  p^+=0 \wedge  n^+=n+1 \wedge
      \mathit{wq}^+ \concat [p^+] = \mathit{wq} \wedge l_0=S \wedge
      l_0^+=P \wedge l_2=C \wedge l_2^+=O\\ \\
      \mathit{out}, n^+=n + 1 \wedge l_0=P \wedge l_0^+=F \wedge l_2=O
      \wedge l_2^+=C\\ \\
      %%%
      \mathit{rq}, \mathit{id}^+=1 \wedge
      \mathit{wq}^+=\mathit{wq} \concat [\mathit{id}^+] \wedge
      l_1=F \wedge l_1^+=N \wedge l_2=C \wedge l_2^+=C\\ \\
      \mathit{stop},
      l_1=N \wedge l_1^+=S\\ \\
      \mathit{go},  p^+=1 \wedge  n^+=n+1 \wedge
      \mathit{wq}^+ \concat [p^+] = \mathit{wq} \wedge l_1=N \wedge
      l_1^+=P \wedge l_2=C \wedge l_2^+=O\\ \\
      \mathit{go},  p^+=1 \wedge  n^+=n+1 \wedge
      \mathit{wq}^+ \concat [p^+] = \mathit{wq} \wedge l_1=S \wedge
      l_1^+=P \wedge l_2=C \wedge l_2^+=O\\ \\
      \mathit{out}, n^+=n + 1 \wedge l_1=P \wedge l_1^+=F \wedge l_2=O
      \wedge l_2^+=C
    \end{array}
    $
  } (x);

  \node[namebox, fit=(closed) (opened) (current bounding box.south west)
  (current bounding box.north east)] (bigsquare) {};

  \node[namebox,anchor=south west] at (current bounding box.north
  west) {$\mathit{Train\ Gate}$};

  \node[namebox,anchor=south east] at
  (bigsquare.north east)
  {$
    \actS = \{\mathit{rq}, \mathit{go}, \mathit{out} \}
    $
  };
\end{tikzpicture}
%%% Local Variables:
%%% mode: latex
%%% TeX-master: "../linearization_through_SOS"
%%% End:
  }
  \caption{Linear version of the gateway
    model}\label{fig:linear-automaton}
\end{figure}

% \todo{Define bisimilarity for LiTS. (Should we do this here or in the
%   section where we define the LiTS?)}

The next step is to prove that the linear version of a composition is
indeed equivalent to it.
If we consider the transition systems they induce, we find that these
differ significantly among each other: they have different labels,
invariants, etc. Thus if we want to prove equivalence at the LiTS
level, we need a non-trivial definition of equivalence.

% Thus it seems it is not possible to prove bisimilarity between a
% composition and its linearized version using these abstract notions
% of bisimilarity.
%
A better strategy is to prove equivalence at the labeled transition
level (using the explicit semantics).
Therefore, we prove that $p$ and its associated linear automaton
% $\alpha_p$
are stateless bisimilar, after abstracting away the values of the
program counters (or location pointers).
%
% Since the definition of L is not guided by the notion of stateless
% bisimilarity I don't think the definition will help the reader in
% understanding the function.
% \attention{We may consider presenting the notion of equivalence we
% have in mind before introducing the Definition of $L$, because the
% actual definition is chosen to have the desired equality.}
The standard notion of strong
bisimilarity~\cite{MousaviRenGro:CongrSOSdataArtInfComp05} is defined
below.

\begin{definition}[Strong Bisimilarity for SOS]
  A symmetric relation $R$ is a strong bisimulation relation if for
  all $(p, q) \in R$, and for all $\sigma$, $\ell$, $p'$, $\sigma'$
  the following transfer conditions hold:
  \begin{enumerate}
  \item $(p, \sigma) \trans{\ell} (p', \sigma') \Rightarrow \langle
    \exists q' :: (q, \sigma) \trans{\ell} (q', \sigma') \wedge (p',
    q') \in R \rangle$
  \item $(p, \sigma) \etrans{\ell} (p', \sigma') \Rightarrow \langle
    \exists q' :: (q, \sigma) \etrans{\ell} (q', \sigma') \wedge (p',
    q') \in R \rangle$
  \end{enumerate}
  Two closed terms $p$ and $q$ are strongly bisimilar, denoted $SOS \models p
  \bisim q$, if $(p,q) \in R$ for some strong bisimulation relation
  $R$.
\end{definition}

Next, we present the SOS rules for the variable scope operator in
Table~\ref{tab:explicit-rules-varscope} (the rules for environment
transitions are similar and therefore omitted). In these rules we make
use of the following notations:
\begin{itemize}
\item Given two sequences $\xs x$ and $\xs y$, such that $\# \xs x =
  \# \xs y$, $\{ \xs x \mapsto \xs y\} \in \ran(\xs x) \rightarrow
  \ran(\xs y)$ is a function defined as follows:
  \begin{equation*}
    \{ \xs x \mapsto \xs y\} = \{ (\xs x {.} i, \xs y {.} i) \ |\
    0 \leq i < \# \xs x\}
  \end{equation*}
\item The notation above is overloaded to denote a similar function.
  We believe this keeps the notation concise and it does not bring
  confusion. Given a sequence $\xs x$ and an element $y$, $\{ \xs x
  \mapsto y\} \in \ran(\xs x) \rightarrow \{ y\}$ is a function
  defined as follows:
  \begin{equation*}
    \{ \xs x \mapsto y\} = \{ (\xs x {.} i,  y) \ |\
    0 \leq i < \# \xs x\}
  \end{equation*}
\item Symbol $\bot$ denotes the undefined value.
\end{itemize}

\begin{table}[htb]
  \centering
  \begin{tabular}{c}
%    \hline \\
    \begin{minipage}[c]{1\linewidth}
      \Sosrule{
        (p, \{ \xs x \mapsto \xs v\} \overwrite \sigma)
        \trans{a, b}
        (p', \{ \xs x \mapsto \xs v'\} \overwrite \sigma')
      }{
        (\vscope{\{ \xs x \mapsto \xs v\}}{p},  \sigma)
        \trans{a, b}
        (\vscope{\{ \xs x \mapsto \xs v'\}}{p'}, \sigma')
      }{rule:action:vscope}
    \end{minipage}\\  \\
    \begin{minipage}[c]{1\linewidth}
      \Sosrule{
        (\vscope{\{ \xs x \mapsto \xs v\}}{p},  \sigma)
        \trans{a, b}
        (\vscope{\{ \xs x \mapsto \xs v'\}}{p'}, \sigma')
      }{
        (\vscope{\{ \xs x \mapsto \bot \}}{p},  \sigma)
        \trans{a, b}
        (\vscope{\{ \xs x \mapsto \xs v'\}}{p'}, \sigma')
      }{rule:action:vscope}
    \end{minipage}\\ \\
%    \hline
  \end{tabular}
  \caption{SOS rules for the variable scope operator}
  \label{tab:explicit-rules-varscope}
\end{table}

Using the previously defined operator and the notion of stateless
bisimilarity, we can enunciate the theorem which states that the
linearization procedure is correct.

\begin{theorem}[Correctness of the Linearization]
  \label{theo:correctness-of-linearization}
  Let $p$ be a composition, and $\lfun{p}=(\alpha_p, \xs l)$. Then we
  have:
  \begin{equation*}
    SOS \models p \bisim
    \vscope{\{ \xs l \mapsto \bot \}}{\alpha_p}
  \end{equation*}
\end{theorem}

\begin{proof}
  It is possible to prove that the following relation:
  \begin{flalign}
    R \triangleq & \{ (p, \vscope{\{\xs l \mapsto \bot\}}{\alpha_{p}}) \
    |\
    (\alpha_p, \xs l) = \lfun{p} \} \cup{} \nonumber\\
    & \{ (p[\xs i], \vscope{\{\xs l \mapsto \xs i\}}{\alpha_{p}}) \ |\
    (\alpha_p, \xs l) = \lfun{p}, \xs i \in \mathrm{locsof}({p}) \}\label{eq:4}
  \end{flalign}
  is a witness of the bisimulation.
  The proof uses the soundness and completeness results presented in
  Sections~\ref{sec:expl-symb-semant} and \ref{sec:relat-lits-sts},
  and it does not require the use of structural induction.
\end{proof}

% \section{Examples}
% \label{sec:examples}

% \todo{These examples will be distributed along the other sections.}

% Next, we present a part of the symbolic transition system. The
% complete STS consists of 16 states. This growth is the result of the
% interleaving actions of the two automata. If more trains were added
% then the size of the model will explode. Note that in this example the
% synchronizing actions do not make the size of the model larger.

\section{Concluding Remarks}
\label{sec:concluding-remarks}

We have presented linearization algorithm for a subset of \CIF, which
shows that every \CIF\ composition can be reduced to an automaton.
The linearization procedure was obtained in a stepwise manner from the
SOS specification of this language.
In this way, SOS rules are used not only to specify the behavior of
\CIF, but also as a specification formalism for performing semantic
preserving manipulations on the syntactic elements of the language.

The soundness and completeness results between the different
transition systems give us a simple proof of correctness on the
linearization procedure. The different levels in which a language is
described (explicit, symbolic, and linear semantics) provide a
convenient way to tackle specific problems. The explicit semantics is
useful for achieving an abstract and succinct specification of the
language. The symbolic semantics give us the means for specifying
symbolic computations. Finally, the linear semantics yields an
efficient representation of the state space associated to a given
composition.

We conjecture the method presented here can be applied to any
automaton based language. For process algebraic specification language
it may not be suitable due to the presence of recursion.

As future work, we plan to extension the linearization algorithm to
the full CIF, and therefore, to a hybrid setting. Time-can-progress
predicates and dynamic types can be extracted in the same way
invariants were extracted in this work, and therefore we expect no
problems in this regard.

%Bibliography
\bibliographystyle{eptcs}
\bibliography{hybchi}

\end{document}